\documentclass[11pt]{article}
\usepackage{graphicx}
\usepackage{amsfonts}
\usepackage{latexsym}
\usepackage{fancyvrb}
\usepackage{booktabs}
\usepackage{amssymb,amsthm,amsmath}

\usepackage{graphicx}
\usepackage[small]{caption}
\usepackage{subcaption}

\usepackage[utf8]{inputenc}
\usepackage{fullpage}
\usepackage{framed}

\usepackage{enumerate}
\usepackage{url}
\usepackage{hyperref}

\usepackage{authblk}
\usepackage[margin=0.5in]{geometry}
\usepackage{graphicx}
\usepackage{xcolor}
\usepackage[capitalise]{cleveref}
	
\usepackage{lineno}

\newtheorem{theorem}{Theorem}
\newtheorem{lemma}{Lemma}

\newtheorem{corollary}{Corollary}
\newtheorem{observation}{Observation}

\newcommand{\etal}{{et~al.}}
\newcommand{\ie}{{i.e.}}

\newcommand{\opt}{\textsf{OPT}}

\newcommand{\ZZ}{\mathbb{Z}} 
\newcommand{\RR}{\mathbb{R}} 
\newcommand{\eps}{\varepsilon}

\def\F{\mathcal F}
\def\I{\mathcal I}
\def\Q{\mathcal Q}

\newcommand{\Area}{\text{Area}}

\def\E{\mathbb{E}}
\def\den{\pi_T}
\def\denL{\pi_L}

\def\eps{\varepsilon}

\def\tri{\triangle}

\theoremstyle{definition}

\newcommand{\later}[1]{}
\newcommand{\old}[1]{}

\title{Lattice and Non-lattice Piercing of Axis-Parallel Rectangles}

\author{
Adrian Dumitrescu\thanks{%
Algoresearch L.L.C., Milwaukee, WI 53217, USA. 
Email~\texttt{ad.dumitrescu@algoresearch.org}.}\qquad
Arsenii Sagdeev\thanks{%
Karlsruhe Institute of Technology, Karlsruhe, Germany.
Email~\texttt{sagdeevarsenii@gmail.com}.}\qquad
Josef Tkadlec\thanks{%
Computer Science Institute, Faculty of Mathematics and Physics, Charles University, Prague, Czech Republic.
Email~\texttt{josef.tkadlec@iuuk.mff.cuni.cz}.}
}

\begin{document}

\maketitle

\begin{abstract}
Given a family $\F$ of shapes in the plane, we study what is the lowest possible density
of a point set $P$ that pierces (``intersects'', ``hits'') all translates of each shape in $\F$.
For instance, if $\F$ consists of two axis-parallel rectangles the best known piercing set,
\ie, one with the lowest density, is a lattice.

Given a finite family $\F$ of axis-parallel rectangles,
we present an algorithm for finding an optimal $\F$-piercing lattice.
The algorithm runs in time polynomial in the number of rectangles
and the maximum aspect ratio of the rectangles in the family.
No prior algorithms for this problem were known.

On the other hand, we show that for every $n \geq 3$, there exists a family of $n$ axis-parallel
rectangles for which the best piercing density achieved by a lattice is separated by
a positive (constant) gap from the optimal piercing density for the respective family.
Finally, we show that the best lattice can be sometimes worse by $20\%$ than the optimal piercing set.


\end{abstract}

\section{Introduction}  \label{sec:intro}

Consider the following problem: if a battleship of a given shape $S$ is hiding in a large sea,
how many bombs are needed to (simultaneously) drop to ensure that at least one bomb hits the battleship?
In other words, what is the smallest possible cardinality of a set of points that pierces each
congruent copy of $S$ that fits within the sea?
To avoid technicalities caused by the boundary effects, such questions are typically studied when
the sea is the whole plane,
and instead of minimizing the cardinality of the piercing set, one minimizes its density. 
The minimum possible density is called the \textit{piercing density} of $S$ and is denoted $\pi(S)$.
It is intimately related to the \textit{covering density} $\vartheta(S)$ of $S$ and is notoriously hard
to compute exactly, apart from several special cases, such as when $S$ is a disk~\cite[Ch.~1]{BMP05}.

Piercing and covering are ubiquitous topics in geometry, and variants of the above problem have
been studied extensively. Here we study the version in which the shape of the battleship is unknown,
but it is known that the shape is one of a few possible options.
Thus, our problem is finding a sparsest point set that pierces each copy of each of the possible shapes.
To keep the situation manageable, following~\cite{DT24,FS89} we consider ships that are rectangular
and axis-aligned, and we only allow translations.

Formally, given a family of sets $\E$ in the Euclidean space, a \emph{piercing set} is a set of points in $\RR^d$
collectively intersecting every set in $\E$. If $\E$ contains an infinite subfamily consisting of pairwise 
disjoint members then clearly every piercing set for $\E$ must be infinite as well. 

For a point set $P$ and a bounded domain $D$ with area $\Area(D)$,
we define the \textit{density of $P$ over $D$} by $\delta(P,D)=\frac{|P\cap D|}{\Area(D)}$.
The \textit{density of $P$} (over the whole plane) is then defined as $\delta(P)=\lim_{r\to\infty} \delta(P,D_r)$,
where $D_r$ is the disk of radius $r$ centered at the origin (if the limit exists).
In particular, the density of a point lattice is the reciprocal of the determinant $\det[u,v]$, where
$[u,v]$ is a vector basis of the lattice~\cite[p.~2]{BMP05}. 

Given a finite family $\F$ of (possibly disconnected) shapes in the plane,
what is the lowest possible \emph{density}
of a point set $P$ that pierces (``intersects'', ``hits'') all translates of each shape in~$\F$?
Regardless of what $\F$ is, such a set must be infinite.  
Here we consider this problem for families $\F$ of axis-parallel rectangles where the goal is finding
a point set of minimum density $\den(\F)$ that collectively pierces each \emph{translate} of a rectangle in $\F$
(the $_T$ in the notation $\den$ is for ``translates''). 
The problem was introduced in~\cite{DT24}, where an approximation algorithm and several other
results have been obtained for such families. 

If $\F$ consists of a single axis-parallel rectangle $R$, an optimal piercing set is a rectangular lattice
given by a tiling of the plane with copies of $R$ (\ie, a covering of the plane by interior-disjoint
translates of $R$).  If $\F$ consists of two axis-parallel rectangles, the answer is unknown!
The best piercing set currently known in this setting, \ie, one with the lowest density, is a lattice:
for certain families of two rectangles, the known lattices are provably optimal whereas for others,
the answer remains elusive~\cite{DT24}. As such, the question posed above appears as a basic open problem.

The starting point of this study is the observation that for certain families $\F$ of
axis-parallel rectangles, there exist sparse $\F$-piercing point sets that are not lattices.
For instance, for the family $\F=\{6 \times 1, 1 \times 6, 3 \times 3\}$ 
there exists a \emph{periodic} $\F$-piercing non-lattice set with the (optimal) density $1/6$;
see~\cref{fig:F0}.  Since the piercing set is  a subset of the integer lattice $\ZZ^2$,
the verification of the piercing property can be done visually and is left to the reader.
Building up on the above observation, we prove a separation result showing that 
a sparsest piercing lattice is sometimes significantly worse than a sparsest piercing point set.
This last step requires a careful setup (see properties P1--P3 below). 

\paragraph{Setup.}
Let $\F=\{R_1,\dots,R_n\}$ be a family of axis-parallel rectangles $R_i = w_i \times h_i$. 
Without loss of generality (after suitable scaling) we may assume that
(i)~the minimum rectangle width and height are both $1$, \ie, $\min w_i=1$ and $\min h_i=1$; and
(ii)~no input rectangle is contained in another. 
Let $k_x = \max\{w_i\}$, $k_y = \max\{h_i\}$, and $k = \max\{k_x,k_y\}$; $k$ is not necessarily integral.
Note that the standard integer lattice is $\F$-piercing.
An \emph{optimal $\F$-piercing lattice} is a sparsest among all $\F$-piercing lattices and its existence follows by the compactness of the set of $\F$-piercing lattices.
Note that an optimal $\F$-piercing lattice is not necessarily unique.
We first show that any optimal $\F$-piercing lattice $\Lambda$, after a possible reflection of the $x$-axis, has a \emph{canonical vector basis} $[u,v]$, i.e., the basis satisfying the following three properties our algorithms rely on:

\begin{enumerate} \itemsep 1pt
\item [P1:] $u$ is a shortest vector in $\Lambda$,
  $u$ has non-negative slope and $v$ has non-positive slope, \ie, $u=[a,b]$, $v=[c,-d]$ with
  $a,b,c,d \geq 0$.
\item [P2:] $1/2 \leq |u| \leq |v| \leq \sqrt{\frac73} \, \frac{A}{\lambda}$,
where $\lambda =|u|$ and $A=a\cdot d+b\cdot c$ is the area of a fundamental parallelogram of $\Lambda$.
\item [P3:] the angle $\alpha$ between $u$ and $v$ is between $45^\circ$ and $135^\circ$.
\end{enumerate} 
  
\paragraph{Outline.}
In Section~\ref{sec:lattice} we show that any optimal $\F$-piercing lattice, after a possible reflection of the $x$-axis, has a canonical vector basis.
For lattices satisfying P1--P3, in Section~\ref{sec:decision} we present a decision algorithm: given
a family $\F$ and a canonical vector basis the decision algorithm determines if the corresponding lattice
is $\F$-piercing. 
In Section~\ref{sec:alg} we present an algorithm for computing an optimal $\F$-piercing lattice.
The algorithm is based on the fact that candidate lattices must be tight (in a sense that will be made precise). 
In Section~\ref{sec:gap} we analyze lattice piercing versus non-lattice piercing
and demonstrate explicit instances for which an optimal piercing lattice is non-optimal
over all piercing sets.

\paragraph{Our results.}
The first result is a decision algorithm. Recall that $k$ is the maximum extent of a rectangle in $\F$, and that the minimum rectangle width and height are both 1.

\begin{theorem} \label{thm:decision}
  Given a family $\F=\{R_1,\dots,R_n\}$ consisting of $n$ axis-parallel rectangles,
 and a canonical vector basis $[u,v]$ of a lattice $\Lambda$, there is an algorithm that determines
 whether $\Lambda$ is an $\F$-piercing lattice in $O(k+n)$ time.
\end{theorem}

The existence of a decision algorithm allows us to obtain an algorithm for computing an optimal piercing lattice.
The algorithm employs the decision algorithm in~\cref{thm:decision} as a subroutine and
computes an optimal piercing lattice by generating and solving linear systems with two variables.

\begin{theorem} \label{thm:alg}
  Given a family $\F=\{R_1,\dots,R_n\}$ consisting of $n$ axis-parallel rectangles,
  an optimal $\F$-piercing lattice can be found in $O\left( k^8 n^4 (k + n) \right)$ time. 
\end{theorem}

In the analysis of lattice piercing versus non-lattice piercing, 
the first distinction is offered by the following. While the resulting separation is small,
its proof is nontrivial~\cite[Section~7]{DT22x}.

\begin{theorem}\label{thm:eps-gap}
There exists a family $\F_0$ of $3$ axis-parallel rectangles and a positive constant $\eps>0$
where the best piercing density achieved by a lattice is at least $(1+\eps) \, \pi_T(\F_0)$. 
\end{theorem}

The proof of~\cref{thm:eps-gap} yields a value $\eps=3 \cdot 10^{-7}$. 
We manage to amplify the separation result by roughly six orders of magnitude
via computer-assisted proofs (Section~\ref{sec:gap}).

\begin{theorem}\label{thm:gap-F0-F1}
There exist families $\F_0,\F_1$ of $3$ (resp. $5$) axis-parallel rectangles for which the best piercing densities
achieved by a lattice are at least $\frac{36}{31} \pi_T(\F_0)$ (resp. $\frac{6}{5} \, \pi_T(\F_1) $).  
\end{theorem}

Theorems~\ref{thm:eps-gap} and~\ref{thm:gap-F0-F1} can be extended to any larger number of rectangles.
Consequently, we obtain the following general separation result.

\begin{corollary} \label{cor:gap}
For every $n \geq 5$,
there exists a family $\F$ of $n$ axis-parallel rectangles for which the best piercing density
achieved by a lattice is at least $\frac{6}{5} \, \pi_T(\F) $. 
\end{corollary}

To put Corollary~\ref{cor:gap} into perspective we highlight the following approximation result
used as a subroutine in our exact algorithm for finding an optimal piercing lattice. 

\begin{theorem} {\rm \cite{DT24}} \label{thm:approx}
Given a family $\F=\{R_1,\dots,R_n\}$ consisting of $n$ axis-parallel rectangles,
a $1.895$-approximation of $\den(\F)$ can be computed in $O(n)$ time.
The output piercing set is a lattice with density at most
$(1+\frac25\sqrt 5)\cdot\den(\F) = (1.894\ldots) \cdot \den(\F)$.
\end{theorem}

Given a point set $P$, a rectangle is said to be \emph{empty} if it contains no points of $P$
in its interior. Let $\F$ be a finite family of axis-parallel rectangles. 
The problem of determining whether a given (infinite) point set $P$ is $\F$-piercing
is closely related to that of determining all \emph{maximal} empty rectangles  amidst the points in~$P$.
As such, a maximal empty rectangle has points of $P$ on its boundary.
A point set $P$ is $\F$-piercing if and only if no maximal empty rectangle $Q$ determined by $P$
strictly contains a translate of a rectangle in $\F$ in its interior (see Lemma~\ref{lem:iff} in \Cref{sec:decision}).

\paragraph{Related work.}
The number of maximal empty rectangles amidst $n$ points is $O(n^2)$, and this bound is tight~\cite{NLH84}.
Given an axis-parallel rectangle $R$ in the plane containing $n$ points, the problem of computing
a maximum-area empty axis-parallel sub-rectangle contained in $R$ has been studied extensively.
Recently, Chan~\cite{Ch21} gave an algorithm running in nearly $O(n \log{n})$ time for this problem.

Let $\F$ be a finite family of axis-parallel rectangles. 
The \emph{piercing number} of $\F$, denoted $\tau(\F)$, is the minimum cardinality of an $\F$-piercing set.
The \emph{independence number} or {\em matching number} of $\F$,
namely the maximum number of pairwise disjoint sets in $\F$,
is denoted by $\alpha(\F)$ or $\nu(\F)$. Clearly, $\nu(\F) \le \tau(\F)$. 
The main unsettled question here is whether $\tau(\F) =O(\nu(\F))$.
The best known upper bound is due to Correa~\etal~\cite{CFPS15}; the above inequality has
been confirmed in some special cases~\cite{CSZ18}.
For axis-parallel boxes in $\RR^d$, Tomon~\cite{To23} has recently shown that the ratio $\tau/\nu$ 
can be arbitrarily large for $d \geq 3$. 

Given a set $S$ of $n$ points in the unit square $U=[0,1]^2$, 
let $A(S)$ be the maximum area of an empty axis-parallel rectangle contained in $U$
(also known as the \emph{dispersion} of $S$), and let $A(n)$
be the minimum value of $A(S)$ over all sets $S$ of $n$ points in $U$.
It is known that $1.504 \leq \lim_{n \to \infty} n A(n) \leq 1.895$. The lower bound
is a recent result of Bukh and Chao~\cite{BC22}
and the upper bound is another recent result of Kritzinger and Wiart~\cite{KW21}.
See also~\cite{AHR17,Kr18,LW23,UV18}.

In the discrete version, the family $\F$ consists of finite (possibly disconnected) subsets of $\ZZ^2$,
and one is interested in the lowest density of an infinite subset of $\ZZ^2$ that
intersects each translate of each element of $\F$ (by an integer vector)~\cite{HLP+22}.
For singleton families in one dimension, \ie, with $\F=\{S\}$  and $d=1$,
the optimal density can be computed in $\exp{(O(k))}$ time, see~\cite{BJR11}.
It is clear that if $|S|=k$, $k \in \{1,2\}$,
then $\pi_T(S)=1/k$. Newman~\cite{New67} showed that $\pi_T(S) \leq 2/5$ for any $S \subset \ZZ$ of size $3$
and this bound is attained if $S=\{0,1,3\}$. Moreover, he showed that $\pi_T(S) \leq (1+o(1)) \frac{\log{k}}{k}$
for any $S \subset \ZZ$ of size $k$, where $k \to \infty$.
 Recently,  Frankl~\etal~\cite{FKS23} confirmed certain explicit
expressions for the piercing densities of $3$-element point sets conjectured by Schmidt and Tuller around 2010. 
Axenovich~\etal~\cite{AGL+19} recently showed via a computer-aided proof that $\pi_T(S) \leq 1/3$ for any
$S \subset \ZZ$ of size $4$, confirming a conjecture of Newman~\cite{New67}.
This bound is attained for $S=\{0, 1, 2, 4\}$.

\subsection{Non-lattice periodic piercing}  \label{sec:non-lattice}

Figure~\ref{fig:F0} shows a non-lattice periodic grid point-set that pierces all translates
of rectangles in $\F_0=\{6 \times 1, 1 \times 6, 3 \times 3\}$. There are six points in a $6 \times 6$ square tile
whose translates can tile the whole plane. This piercing set yields an area per point equal to $36/6=6$,
so its piercing density is $1/6$. Since the minimum area of a rectangle in $\F$ is $6$, this piercing set
has an optimal density (by averaging over a region of large area). 

\begin{figure}[ht]
\centering
\includegraphics[scale=1]{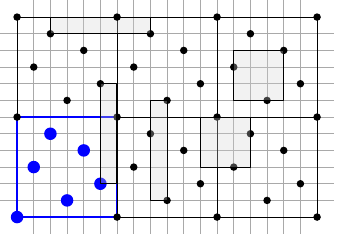}
\caption{This non-lattice periodic grid point-set pierces all translates
  of $\F_0=\{6 \times 1, 1 \times 6, 3 \times 3\}$.} 
\label{fig:F0}
\end{figure}

\section{Lattice piercing}  \label{sec:lattice}

Let $\F$ be a finite family of axis-parallel rectangles; denote by $A_{\min}$
the minimum area of a rectangle in $\F$.
Let $\Lambda$ be an optimal $\F$-piercing lattice and $A>0$ be the area of its fundamental parallelogram. 

\begin{observation} \label{obs:min-area}
$A \leq A_{\min}$. In particular, $A \leq k$.
\end{observation}
\begin{proof}
  The first inequality is clear. For the second, recall that there is a rectangle $1 \times h_i$
  in $\F$, where $h_i \leq k$. 
\end{proof}

Two parameters of interest in a lattice are (see also~\cite[Ch.~4]{ES03}):
\begin{itemize} \itemsep 0pt
\item the smallest interpoint distance $\lambda$
\item the distance $\mu$ between two consecutive lattice lines of direction $u$,
  with  interpoint distance $\lambda$; these lines are also referred to as $u$-lines.
\end{itemize}
Note that $A=\lambda \mu$.
We start with a lemma of independent interest spelling out two inequalities between these parameters.
The second can be found in~\cite{PA95} (as Thm.~1.3).  Nevertheless, here we give our own proof.

\begin{lemma} \label{lem:lambda}
The following inequalities hold:
{\rm (i)}~$\lambda \leq \frac{2}{\sqrt3} \mu$ and
{\rm (ii)}~$\lambda \leq \sqrt{\frac{2 A}{\sqrt3}}$.
\end{lemma}
\begin{proof}
Consider two consecutive parallel $u$-lines $\ell_1, \ell_2$ at distance $\mu$ from each other.
Refer to \cref{fig:lambda}.

\begin{figure}[ht]
\centering
\includegraphics[scale=0.7]{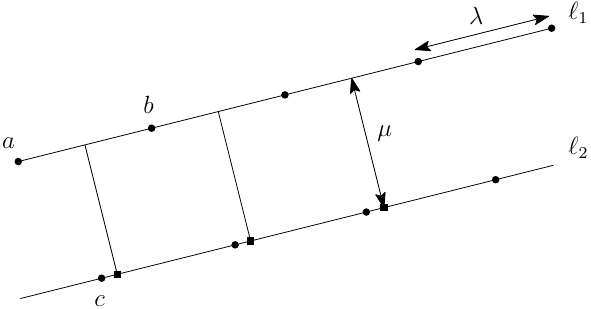}
\caption{Two consecutive parallel $u$-lines.}
\label{fig:lambda}
\end{figure}

Let $a,b \in \Lambda \cap \ell_1$ be consecutive lattice points on $\ell_1$ and
$c \in \Lambda \cap \ell_2$ be a lattice point whose orthogonal projection onto $\ell_1$
lies between $a$ and $b$, and say, $|ca| \leq |cb|$. Let $c'$ be a point on $\ell_2$ whose orthogonal projection onto $\ell_1$ is the midpoint of $ab$.
Observe that $\lambda \leq |ca| \leq |c'a| = \sqrt{\lambda^2/4 + \mu^2}$,
or $\lambda \leq \frac{2 \mu}{\sqrt3}$, proving the first inequality.
Taking into account that $A=\lambda \mu$, we deduce
$ \lambda^2 \leq \frac{2 \lambda \mu}{\sqrt3} = \frac{2A}{\sqrt3}$, 
or $\lambda \leq \sqrt{\frac{2 A}{\sqrt3}}$, proving the second inequality.
\end{proof}

\paragraph{Canonical vector basis.}
We show that an optimal $\F$-piercing lattice $\Lambda$ has a canonical vector basis $[u,v]$, \ie, one satisfying properties P1--P3 listed in Section~\ref{sec:intro}. 
Let $u$ be a shortest vector in $\Lambda$ ($|u|=\lambda$).
We may assume without loss of generality (this may require a reflection of the $x$-axis)
that $u$-lines have non-negative slope. 
Let $\ell$ be a $u$-line through the origin $o$ and $\ell'$ be the next $u$-line below $\ell$.
Refer to Fig.~\ref{fig:angle}. Observe that $[u,v]$ is a basis of $\Lambda$ for every $v = oc$, $c \in \Lambda \cap \ell'$.
Our next lemma shows that $v$ satisfies the desired properties P1--P3 for a suitable choice of $c$.

\begin{lemma} \label{lem:angle}
  There exists $c \in \Lambda \cap \ell'$ such that $v=oc$ has non-positive slope, the angle $\alpha$ made by $u$ and $v$ satisfies
  $45^\circ \leq \alpha \leq 135^\circ$, and $|v| \leq \sqrt{\frac73} \, \frac{A}{\lambda}$.
\end{lemma}
\begin{proof}
	
  The easy case when $\ell$ and $\ell'$ are horizontal or vertical is omitted from the proof. We will find the desired point $c$ in the fourth quadrant:
  since $\mu \geq \frac{\sqrt3}{2} \lambda$ (by Lemma~\ref{lem:lambda}), it is easily seen that there is at least one lattice point in the fourth quadrant on $\ell'$.
Let $a$ be the endpoint of $u$ and let $p$ denote the orthogonal projection of $o$ onto $\ell'$. We distinguish two cases:

\smallskip
\emph{Case 1: the slope of $\ell$ is at most $1$.} See Fig.~\ref{fig:angle}\,(left).
Let $c$ be the closest lattice point left of $p$ on $\ell'$, if there is one in the fourth quadrant,
or the first one right of $p$, otherwise. 
By the assumption on slope we have $\angle{aoc} \leq 135^\circ$, as required.
If $c$ lies left of $p$, then clearly we have $\angle{aoc} \geq \angle{aop} = 90^\circ$. 
If $c$ lies right of $p$, it must be the first point on $\ell'$ in the fourth quadrant. 
Since $ca$ has positive slope, we have $\angle{oac} \leq 90^\circ$, whence
$\angle{aoc} + \angle{oca} \geq 90^\circ$.
In $\tri oac$ the side $oa$ is the shortest one, hence the corresponding angle $\angle{oca}$ is the smallest one.
In particular $\angle{oca} \leq \angle{aoc}$ hence $\angle{aoc} \geq 45^\circ$.
Finally, recall that $A=\lambda \mu$ and that $\lambda \leq \frac{2 \mu}{\sqrt3}$ (by the first inequality of Lemma~\ref{lem:lambda}).
Thus, we have
\[ |v| = |oc| =\sqrt{\mu^2 + |pc|^2} \leq \sqrt{\mu^2 + \lambda^2} \leq \sqrt{\mu^2 + \frac{4\mu^2}{3}} = \sqrt{\frac73} \, \mu
= \sqrt{\frac73} \, \frac{A}{\lambda}. \]

\begin{figure}[ht]
	\centering
	\includegraphics[scale=0.7]{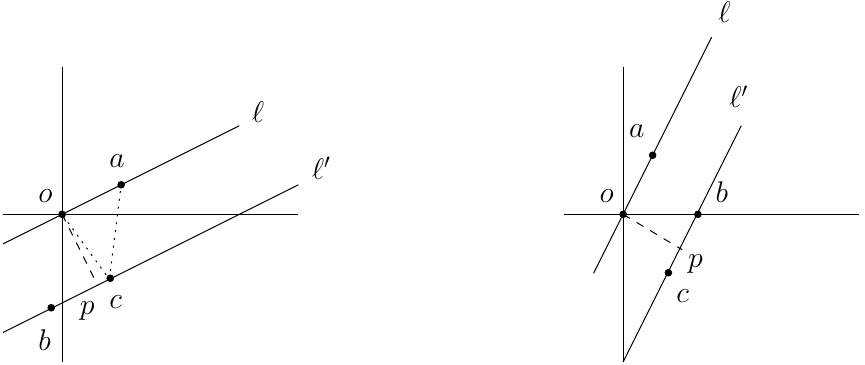}
	\caption{The angle between $u$ and $v$ is mid-range.
	Case 1 (left) and Case 2 (right).}
	\label{fig:angle}
\end{figure}

\smallskip
\emph{Case 2: the slope of $\ell$ is at least $1$.} 
Let $c$ be the closest lattice point right of $p$ on $\ell'$, if there is one in the fourth quadrant, or the closest one left of $p$, otherwise.
By the assumption on slope we have $\angle{aoc} \geq 45^\circ$, as required. If $c$ lies right of $p$, then clearly we have
$\angle{aoc} \leq \angle{aop} = 90^\circ$. If $c$ lies left of $p$, it must be the last point on $\ell'$ in the fourth quadrant.
Let $b$ be the closest lattice point right of $p$ on $\ell'$, which is in the first quadrant.
Since $ob$ has positive slope, we have $\angle{aob} \leq 90^\circ$, whence
$\angle{oab} + \angle{oba} \geq 90^\circ$.
In $\tri oab$ the side $oa$ is the shortest one, hence the corresponding angle $\angle{oba}$ is the smallest one.
In particular $\angle{oba} \leq \angle{oab}$ hence $\angle{oab} \geq 45^\circ$ or $\angle{aoc} \leq 135^\circ$ (as complementary angles).
Finally, the inequality $|v| \leq \sqrt{\frac73} \, \frac{A}{\lambda}$ follows as before.
\end{proof}

To verify the optimality of the interval $45^\circ \le \alpha \le 135^\circ$ from \Cref{lem:angle}, consider two lattices obtained from $\ZZ^2$
by a small perturbation as in Fig.~\ref{fig:angles-45-135}.

\begin{figure}[ht]
\centering
\includegraphics[scale=0.65]{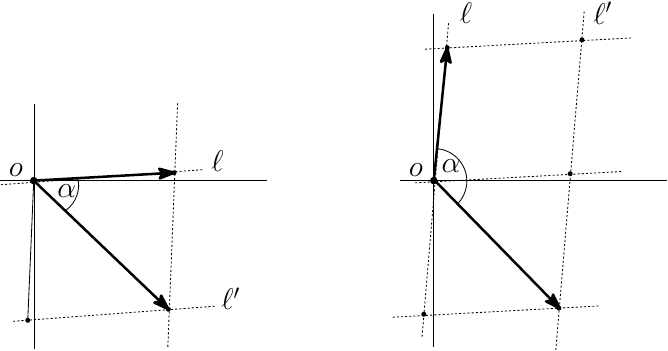}
\caption{Tight cases of optimality.}
\label{fig:angles-45-135}
\end{figure}

The property that $\alpha$ is not too far from $90^\circ$
is key in bounding from above the complexity (and length) of a path connecting any two points
in the lattice along directions $u$ and $v$. (See also~\cite[Ch.~27]{Va01} for a similar angle property
of the shortest vector derived from different principles.)
Another key property for this purpose is that $|u| \geq 1/2$, which holds by the following lemma.

\begin{lemma} \label{lem:uv}
  If $\Lambda$ is an optimal $\F$-piercing lattice, then $|u|\ge 1/2$.
\end{lemma}
\begin{proof}
  Assume for contradiction that $|u| < 1/2$. 
  Recall that each extent of a rectangle in $\F$ is at least $1$, and so the integer lattice $\ZZ^2$
  is $\F$-piercing and thus $A=\lambda \mu \geq 1$. Since $\lambda = |u| < 1/2$, we have  $\mu>2$.
  Pick a positive $\varepsilon$ such that $(2+2\varepsilon)\lambda < 1$ and consider a new lattice $\Lambda'$ with basis $[u',v'] = [(2+2\varepsilon)u,v/2]$. 
  Observe that
  \[ \lambda' = (2+2\varepsilon)\lambda <1, \ \mu'= \mu/2 >1, \text{ and } A'= (1+\varepsilon)\lambda\mu > A. \]
  Next, we show that $\Lambda'$ is also $\F$-piercing and thus $\Lambda$ is not optimal, a desired contradiction.
  
  Let $R$ be an arbitrary translate of some $R_i \in \F$. Our goal is to show that $R$ is pierced by $\Lambda'$.
  Since $\Lambda$ is $\F$-piercing, $R$ is intersected by some $u$-line, say, $\ell_1$.
  Observe that if $|\hat{\ell} \cap R| \geq 1 > \lambda'$ for some $u'$-line $\hat{\ell}$ of $\Lambda'$, then $R$ is pierced by a lattice point
  on $\hat{\ell}$, as desired. So in what follows, we assume that $|\hat{\ell} \cap R|<1$ for each $u'$-line $\hat{\ell}$.
  In particular, no $u'$-line intersects $R$ in two opposite sides. Thus, we can assume without loss of generality that
  $\ell_1$  intersects the left and top side of $R$. In particular, $\ell_1$ is not axis-parallel.
  
  Let $0<\theta<90^\circ$ denote the angle made by a $u$-line $\ell_1$ with the $x$-axis. 
  Let $\ell_2$ be the next $u$-line of $\Lambda$ below it, and $\ell'$ be the $u'$-line of $\Lambda'$ between $\ell_1$ and $\ell_2$. See Fig.~\ref{fig:halve}.
  Let $y$ denote the distance between the upper-left corner of $R$ and $\ell_1$;
  and let the orthogonal projection of this corner onto $\ell_1$ partition $\ell_1 \cap R$ into segments of lengths $s_1$ and $s_2$.
  Writing $t= \tan{\theta}$, we have
  \[ 1 > |\ell_1 \cap R| = s_1 + s_2 = y (t +t^{-1}), \]
  and thus $y<1/2$ since $t + t^{-1} \geq 2$ for every $t>0$.
  
  Since $\Lambda$ is $\F$-piercing, the distance from $R$ to $\ell_2$, say, $y'$, is at most $y$ (otherwise,
  a translate of $R$ that fits perfectly in between $\ell_1$ and $\ell_2$ is not pierced, a contradiction). In particular, $y'\le y < 1/2 < \mu'$,
  and thus $\ell'$ intersects $R$. We distinguish two cases.

\begin{figure}[ht]
\centering
\includegraphics[scale=0.77]{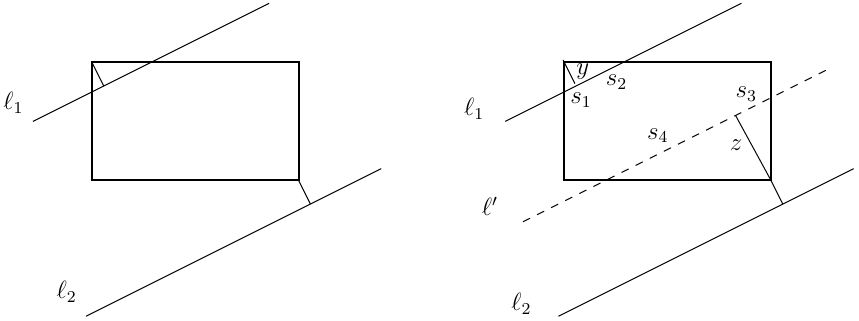}   
\caption{Left: a translate $R$ of $R_i \in \F$ is pierced by $\Lambda$. 
  Right: the same translate is pierced by $\Lambda'$.}
\label{fig:halve}
\end{figure}

\smallskip
\emph{Case 1:} $\ell'$ intersects the right and bottom side of $R$.
Let $z$ denote the distance between the lower-right corner of $R$ and  $\ell'$; and let the orthogonal projection of this corner
onto $\ell'$ partition $\ell' \cap R$ into segments of lengths $s_3$ and $s_4$.
We have
\[z \ge \mu'-y' > 1-y' > \frac{1}{2} \text{ thus } |\ell' \cap R| = s_3 + s_4 = z (t +t^{-1}) > \frac12 \cdot 2 =1 > \lambda'.\]
This means that $R$ is pierced by a point of $\Lambda'$ on $\ell'$, as desired.  

\smallskip
\emph{Case 2:} $\ell'$ intersects the left and top side of $R$.
Let $z$ denote the distance between the upper-left corner of $R$ and  $\ell'$; and let the orthogonal projection of this corner onto $\ell'$
partition $\ell' \cap R$ into segments of lengths $s_3$ and $s_4$.
We have
\[z = \mu'+y > \mu' > 1 \text{ thus } |\ell' \cap R| = s_3 + s_4 = z (t +t^{-1}) > 2 > \lambda'.\]
This again means that $R$ is pierced by a point of $\Lambda'$ on $\ell'$, as desired.  
\end{proof}

In summary, for any finite family $\F$ of axis-parallel rectangles of dimensions at least $1$, any optimal $\F$-piercing lattice $\Lambda$,
after a possible reflection of the $x$-axis, has a canonical vector basis $[u,v]$ satisfying properties P1--P3 given in Section~\ref{sec:intro}.

\section{Decision algorithm} \label{sec:decision}

\begin{proof}[Proof of Theorem~\ref{thm:decision}]
  Let $\F=\{R_1,\dots,R_n\}$ be a family of $n$ axis-parallel rectangles $R_i = w_i \times h_i$.
  Recall that $k_x = \max\{w_i\}$, $k_y = \max\{h_i\}$, and $k = \max\{k_x,k_y\}$.
  Let $[u,v]$ be a canonical vector basis of a lattice $\Lambda$. Fix $m_x=k_x+1$, $m_y=k_y+1$ and $m=\max\{m_x, m_y\}$. (Note that any other choice of $m_x$ and $m_y$ slightly larger than $k_x$ and $k_y$, respectively, works just as well.) 
  For the \emph{board rectangle} $D=[0,m_x] \times [-m_y,m_y]$, let $\Q$ be the set of all maximal empty axis-parallel sub-rectangles of $D$ whose left sides contain the origin $o$.
  The set $\Q$ determines whether $\Lambda$ is $\F$-piercing as follows.

\begin{lemma} \label{lem:iff}
	A lattice $\Lambda$ is $\F$-piercing if and only if for every
	$Q = w \times h \in \Q$ and $R_i = w_i \times h_i \in \F$, we have
	$ w  \leq w_i \text{ or } h \leq h_i.$
\end{lemma}
\begin{proof}
	For the direct implication, we assume that $\Lambda$ is $\F$-piercing; further assume for contradiction that
	$w>w_i$ and $h>h_i$ for some $Q \in \Q$ and $R_i \in \F$. It follows that a translate of $R_i$ can be placed
	strictly inside $Q$. As such, this translate is not pierced by $\Lambda$, a contradiction.
	
	For the converse implication, we assume that some translate of $R_i = w_i \times h_i \in \F$ is not pierced by $\Lambda$. It follows that there exists
	a slightly larger scaled copy $R'_i = w'_i \times h'_i$ of $R_i$ (concentric with the translate of $R_i$) that is not pierced by~$\Lambda$.
	Let $Q' = w' \times h'$ be a maximal empty rectangle containing $R'_i$. Note that either the left side of $Q'$ contains a point $p$ of $\Lambda$, or $Q'$ has infinite width and thus the bottom side of $Q'$ contains a point $p$ of $\Lambda$.
	Since  $\Lambda$ is periodic, we can assume without loss of generality that $p=o$. Consider any $Q = w \times h \in \Q$  that contains the empty rectangle $Q'\cap D$. Observe that $w \ge \min\{w',m_x\}>w_i$ and $h \ge \min \{h', m_y\}>h_i$. In particular, neither $w  \leq w_i$ nor $h \leq h_i$ holds, as desired.
\end{proof}

\begin{lemma} \label{lem:funnel}
  There are $O(k)$ rectangles in $\Q$ and they can be found in $O(k)$ time.
\end{lemma}
\begin{proof}
  Let $P$ be the set of lattice points $p$ in the interior of $D$ or in the interior of its left side such that the rectangle with diagonal $op$ is empty and contains no other lattice points on its boundary.
  Let $P= P^+ \cup P^-$ be the partition of $P$ induced by the $x$-axis
  ($P^+$ and $P^-$ are the points of $P$ with non-negative and non-positive $x$-coordinates, respectively.)
  Order $P^+$ and $P^-$ by $x$-coordinate and observe that:
  (i)~$P^+$ must form a sequence of points with decreasing $y$-coordinates; and 
  (ii)~$P^-$ must form a sequence of points with increasing $y$-coordinates.
  We refer to $P= P^+ \cup P^-$ as the \emph{funnel}, where $P^+$ is the
  \emph{upper funnel} and $P^-$ is the \emph{lower funnel}.

\vspace{2mm}

\emph{Generating the funnel.}
Recall that $u$ lies in the first quadrant and $v$ in the fourth quadrant.
For a fixed $j \in \ZZ$, let $\ell_j$ be the lattice line consisting of the lattice points $i \cdot u + j \cdot v$,
with $i \in \ZZ$. Since $u \in \ell_0$ yields the shortest interpoint distance in $\Lambda$, it is the only point on $\ell_0$ that can belong to $P^+$. We include $u$ in  $P^+$ if $u$ lies in the interior of $D$ or in the interior of its left side.
Then we process the $u$-lines $\ell_j$ one by one: lines with $j<0$ are processed in  decreasing order of $j$;
lines with $j>0$ are processed in  increasing order of $j$. 

When processing the line $\ell_j$, $j<0$, we compute the point with smallest non-negative $x$-coordinate and include it
in $P^+$ if its $x$-coordinate is less than the $x$-coordinates of the points already processed and if it lies in the interior of $D$ or in the interior of its left side.
(Indeed, the new point is always higher than already processed ones since the minimum $i$ such that $i\cdot a+j\cdot c\ge 0$ is monotone as a function of $j$.)
These points of $P^+$ form an upward \emph{branch} approaching the positive part of the $y$-axis.
We stop generating the branch once any of these occur:
(i)~the $u$-line is ``too high'' (it no longer intersects the board) or
(ii)~we get a point on the $y$-axis.

Similarly, we add to $P^+$ points that form a branch approaching the positive part of
the $x$-axis; these points are obtained by processing the lines $\ell_j$, $j>0$.
The points in $P^+$ (the union of two branches) are naturally sorted by their $x$- and $y$-coordinates,
as they form the upper funnel. Refer to \cref{fig:funnel}. 

\begin{figure}[ht]
\centering
\includegraphics[scale=0.57]{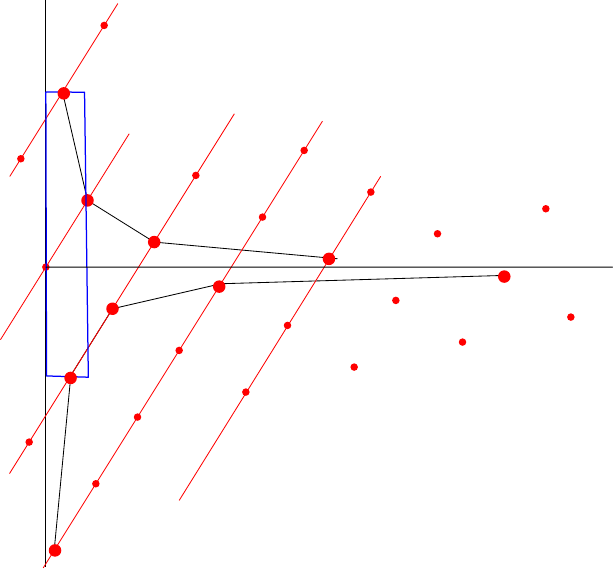}
\caption{A maximal empty rectangle supported from the left by the origin amidst the points of the funnel.}
\label{fig:funnel}
\end{figure}
 
The computation of the lower funnel is analogous to the computation of the upper funnel. 
As before, we construct two branches, this time approaching the positive part of the $x$-axis and the negative
part of the $y$ axis. The points in $P^-$ are likewise naturally sorted by their $x$- and $y$-coordinates,
as they form the lower funnel.

Note that if a lattice point $i \cdot u + j \cdot v$ is in $D$, then
by the law of sines, one can show that
\begin{equation} \label{eq:ij}
	|i|+|j| \leq \frac{\sqrt{m_x^2+m_y^2}}{\lambda \sin 22.5^\circ} \le  \frac{2\,\sqrt{m_x^2+m_y^2}}{\sin 22.5^\circ} < 5.23\,\sqrt{m_x^2+m_y^2} < 7.4\, m = 7.4(k+1).
\end{equation}
Therefore, there are only $O(k)$ lines to be processed, each in $O(1)$ time.
Thus $P$ is of size $O(k)$ and can be computed in $O(k)$ time.

\vspace{2mm}

\emph{Generating the set $\Q$.}
Observe that every pair of consecutive points from $P^+$ with non-zero $x$-coordinates determines an empty rectangle in $\Q$:
it is bounded by the highest of them from above, the rightmost of them from the right, the origin from the left, and continues down
until it hits either some point from $P^-$ with non-zero coordinates or the bottom side of $D$. 
Note that the widths of these rectangles increase and their heights decrease as we process the points of $P^+$ from left to right.
Every pair of consecutive points from $P^-$ with non-zero $x$-coordinates determines an empty rectangle in $\Q$ in a similar fashion.

There are at most 4 other rectangles in $\Q$ which we informally describe as ``the widest'' and ``the highest'' ones. First, we give a precise definition of ``the highest'' rectangles. Let $Q_y^+$ be the rectangle bounded by the $y$-axis from the left, the top side of $D$ from above, the point of $P^+$ with the smallest non-zero $x$-coordinate (if there are no points with non-zero $x$-coordinates in $P^+$, then by the right side of $D$) from the right, and continues down until it hits either some point from $P^-$ with non-zero coordinates or the bottom side of $D$.
We define $Q_y^-$ in a similar fashion after switching the roles of $P^+$ and $P^-$. 

We claim that $Q_y^+$ is empty. We only show that it contains no lattice points from the first quadrant in the interior, since the proof for the fourth quadrant is analogous. If $P^+$ does not contain a point on the $y$-axis,
then $Q_y^+$ is empty by definition of the funnel. Otherwise, if $[0,y_0]\in P^+$ but $Q_y^+$ is not empty,
consider the point $[x_1,y_1]$ of $\Lambda$ with the smallest non-negative $y$-coordinate in the interior of $Q_y^+$.
On the one hand, we have $y_1>y_0$ by definition of the funnel. On the other hand, since $\Lambda$ is periodic,
we have $[x_1,y_1-y_0] \in \Lambda$, which contradicts the choice of $[x_1,y_1]$. Thus $Q_x^+$ is empty, as claimed, and so $Q_y^+ \in \Q$. A similar argument shows that $Q_y^- \in \Q$ as well.

If there is a point of the funnel on the $x$-axis, then there are 2 more (``the widest'') rectangles in $\Q$, each of width $m_x$. First, the rectangle $Q_x^+$ bounded by the $y$-axis from the left, the $x$-axis from below, the right side of $D$ from the right, and the point of $P^+$ with the smallest non-zero $y$-coordinate (if there are no points with non-zero $y$-coordinates in $P^+$, then by the top side of $D$) from above. Second, the rectangle $Q_x^-$ defined in the same fashion for the fourth quadrant. Otherwise, if there are no points of the funnel on the $x$-axis, then there is a unique ``the widest'' rectangles $Q_x$ in $\Q$ of width $m_x$: it is bounded by the $y$-axis from the left, the right side of $D$ from the right, the point of $P^+$ with the largest non-zero $x$-coordinate (if there are no points with non-zero $x$-coordinates in $P^+$, then by the top side of $D$) from above, and the point of $P^-$ with the largest non-zero $x$-coordinate (if there are no points with non-zero $x$-coordinates in $P^-$, then by the bottom side of $D$) from below. We omit the proof that these rectangles are maximal empty, i.e., belong to $\Q$, since it is very similar to the corresponding proof for $Q_y^+$ above.

To illustrate the above procedure, observe that if $u=[0,1], v=[1,0]$, then $P^+=\{[0,1], [1,0]\}$, $P^-=\{[0,-1], [1,0]\}$,
and the only rectangles in $\Q$ are $Q_x^+=[0,m_x]\times[0,1]$, $Q_x^-=[0,m_x]\times[-1,0]$ and $Q_y^+=Q_y^-=[0,1]\times[-m_y,m_y]$.

Note that the above procedure always generates all rectangles in $\Q$.
Besides, it finds $O(k)$ rectangles in $O(k)$ time, as desired.
\end{proof}

\emph{Decision algorithm and its complexity.}
The algorithm consists of two phases.
First, all $O(k)$ rectangles in $\Q$ are found in $O(k)$ time as in \Cref{lem:funnel}.
Then, we test the conditions of Lemma~\ref{lem:iff} for these rectangles to determine whether the lattice $\Lambda$ is $\F$-piercing.

Specifically, we consider rectangles in $\F$ one by one. Given $R_i = w_i \times h_i \in \F$, let $Q_j =  w_j' \times h_j'$ be a rectangle in $\Q$
with the smallest width such that $w_j'>w_i$. Note that to test the conditions of \Cref{lem:iff} for $P_i$ and all rectangles in $\Q$,
it is sufficient to check if $h_j'\le h_i$.

Observe that $w_j'$ is monotone as a function of $w_i$. Moreover, we can assume without loss of generality that both sets $\F$ and $\Q$
are sorted in increasing order of width. (Indeed, the set $\F$ is given while the set $\Q$ is generated in this order during the first phase.)
Consequently, to find the value of $w_j'$ in the $i$th step, there is no need to examine the rectangles $Q \in \Q$ that were too narrow
in previous steps. More formally, we can  increment the index $j$ until the inequality $w_j'>w_i$ holds for the first time.
Since both indices $i$ and $j$ cannot decrease during this phase, its running time is $O(|\Q|+|\F|) = O(k+n)$.

Note that the time taken by both phases is $O(k+n)$. This completes the proof of Theorem~\ref{thm:decision}.
\end{proof}

\section{Computing an optimal $\F$-piercing lattice}  \label{sec:alg}

In this section we prove Theorem~\ref{thm:alg}.
Given a set $\F$ of $n$ axis-parallel rectangles, we obtain a pseudo-polynomial time algorithm for finding an optimal $\F$-piercing lattice.
Its running time is polynomial in the number of rectangles and the maximum dimension of the rectangles in $\F$.

\subsection{Preliminaries}

By properties P1--P3, the area of a fundamental parallelogram of $\Lambda$ is $A=a\cdot d+b\cdot c>0$.
As such, $A$ is an increasing function in each of the four variables. 

A lattice $\Lambda$ is said to be $x$-\emph{tight} with respect to $\F$ and its basis $[u,v]$ if for some
$i_s,j_s,i_t,j_t \in \ZZ$ satisfying~\eqref{eq:ij} and some $w_s, w_t$ each of which is either $0$ or $m_x$ (recall that $m_x=k_x+1$, where $k_x$ is the width of the  widest rectangle in $\F$) or a width of a rectangle in $\F$, we have
\begin{align*}
i_s \cdot a + j_s \cdot c &= w_s,\\
i_t \cdot a + j_t  \cdot c &= w_t,
\end{align*} 
where these constraints are linearly independent, \ie, $i_s \cdot j_t- j_s \cdot i_t \neq 0$.

A lattice $\Lambda$ is said to be $y$-\emph{tight} with respect to $\F$ and its basis $[u,v]$ if
for some
$i_s,j_s,i_t,j_t \in \ZZ$ satisfying~\eqref{eq:ij} and some $h_s, h_t$ each of which is either $0$ or $m_y$ (recall that $m_y=k_y+1$, where $k_y$ is the height of the  highest rectangle in $\F$) or a height of a rectangle in $\F$, we have
\begin{align*}
	i_s \cdot b + j_s \cdot d &= h_s,\\
	i_t \cdot b + j_t  \cdot d &= h_t,
\end{align*} 
where these constraints are linearly independent.

A lattice $\Lambda$ is said to be \emph{tight} with respect to $\F$ and its basis $[u,v]$
if $\Lambda$ is both $x$-tight and $y$-tight with respect to $\F$ and $[u,v]$. 
The key to the algorithm is the following.

\begin{lemma} \label{lem:tight}
  There exists an optimal $\F$-piercing lattice that is tight with respect to $\F$ and a canonical basis.
\end{lemma}
\begin{proof}
  Let $\Lambda$ be an optimal $\F$-piercing lattice of density $1/A$ and $[u,v]$ be its canonical basis. 
  Consider the set $\Q$ of rectangles constructed in \Cref{lem:funnel}. By \Cref{lem:iff},
  for every $Q=w\times h \in \Q$ and $R_i = w_i \times h_i \in \F$, either $w \leq w_i$ or $h \leq h_i$ holds.

  Each satisfied constraint of the form $h \leq h_i$ can be expressed as a linear inequality in $b$ and $d$.
  Indeed, if $Q$ is not supported by lattice points on the top or the bottom side, then $h\ge m_y$,
  and so $h \leq h_i$ cannot hold.
  Thus we can assume that $Q$ is supported by lattice points on both the top and bottom sides.
  Since the $y$-coordinates of these points can be expressed as integer combinations of $b$ and $d$, then so is $h$.
  Moreover, the coefficients of this expression satisfy~\eqref{eq:ij} because the difference between the rightmost and the leftmost
  of these two lattice points is also a point of $\Lambda$ in the board rectangle $D$.
  Let $\I_y$ be a finite set of all these inequalities in $b$ and $d$.

Suppose now that we continuously change the values of $b$ and $d$ such that all the inequalities in $\I_y$ are satisfied
at any given moment. Naturally, the sets $\Q$ and $P$ constructed in \Cref{lem:funnel} also change continuously. Then the
decision algorithm from \Cref{sec:decision} implies that the lattice remains $\F$-piercing during the process unless the
`structure' of the funnel $P$ changes, which, in turn, can only happen if (i) a new lattice point reaches the $x$-axis or one of the horizontal sides
of $D$; or (ii) a lattice point on the $x$-axis or one of the horizontal sides of $D$ leaves this axis or side. 

Note that the structure of $P$ can change as well when $y$-coordinates of two lattice points in $D$ become equal or, vise versa, two lattice points in $D$ on the same horizontal line move with different speeds. However, since $\Lambda$ is periodic, these events imply the two listed above. 

Each non-zero lattice point $i \cdot u - j \cdot v$ in $D$ on the $x$-axis corresponds to the equality $i \cdot b + j \cdot d = 0$.
Moreover, each lattice point $i \cdot u - j \cdot v$ with the $x$-coordinate in the interval $[0,m_x]$
that lies above (resp., below) the $x$-axis corresponds to the strict inequality $i \cdot b + j \cdot d > 0$ (resp.,
$i \cdot b + j \cdot d < 0$). Define similar expressions corresponding to the relative positions of lattice points and the
horizontal sides of $D$. Let $\I_y'$ be the set of all these expressions. Observe that if we continuously change the
values of $b$ and $d$ such that all the expressions in $\I_y \cup \I_y'$ remain satisfied, then the structure of the set
$P$ does not change and the lattice remains $\F$-piercing. Moreover, the right hand side of each of the expressions in
$\I_y \cup \I_y'$ is either $0$ or $m_y$ or a height of a rectangle in $\F$.
 
We claim that there is at least one equality in $\I_y \cup \I_y'$. Indeed, otherwise we could slightly increase both $b$
and $d$ such that all expressions in $\I_y \cup \I_y'$ remain satisfied. The new lattice would still be $\F$-piercing
and have area of the fundamental parallelogram larger than $A$, and thus $\Lambda$ would not be optimal. This
contradiction proves the claim. If there is another equality in  $\I_y \cup \I_y'$ linearly independent with the first
one, then $\Lambda$ is $y$-tight with respect to $\F$ and $[u,v]$. 

Similarly, define the sets $\I_x$ and $\I_x'$ of linear expressions in $a$ and $c$ such that the lattice remains
$\F$-piercing as we continuously change the values of $a$ and $c$ provided that all the expressions in $\I_x \cup \I_x'$
remain satisfied. As before, there is at least one equality in $\I_x \cup \I_x'$, and if there is another equality in
$\I_x \cup \I_x'$ linearly independent with it, then $\Lambda$ is $x$-tight with respect to $\F$ and $[u,v]$. 

We claim that $\Lambda$ is either $x$-tight or $y$-tight with respect to $\F$ and $[u,v]$. Indeed, assume for
contradiction that $i_t \cdot a + j_t \cdot c = w_t$ and $i_s \cdot b + j_s \cdot d = h_s$ are the only (up to scaling)
equalities in $\I_x \cup \I_x'$ and $\I_y \cup \I_y'$, respectively. Take 
\begin{equation*}
	a(\tau) = a+j_t\cdot \tau, \ \ c(\tau)=c-i_t \cdot \tau, \ \ b(\sigma) = b+j_s\cdot \sigma, \ \ d(\sigma)=d-i_s \cdot \sigma
\end{equation*}
and  consider the area of the fundamental parallelogram as a bilinear function of $\sigma, \tau$:
\begin{equation*}
	A(\sigma, \tau) = a(\tau)\cdot d(\sigma)+b(\sigma)\cdot c(\tau) = r_{s,t}\cdot \sigma \cdot \tau + r_{s}\cdot \sigma + r_{t}\cdot \tau + r, \mbox{ where}
\end{equation*}
\begin{equation*}
	r_{s,t} = -(i_s\cdot j_t + j_s \cdot i_t), \ \  r_s = j_s \cdot c - i_s \cdot a, \ \  r_t = j_t\cdot d - i_t \cdot b, \ \  r=a\cdot d+b\cdot c = A.	
\end{equation*}
On the one hand, it is not hard to see that if at least one of the coefficients $r_{s,t}, r_{s}, r_{t}$ is not zero,
then we can pick sufficiently small $\sigma$ and $\tau$ such that $A(\sigma,\tau)>A$ and all the expressions in $\I_x
\cup \I_x'$ and $\I_y \cup \I_y'$ remain satisfied. The new lattice is still $\F$-piercing, and thus $\Lambda$ is not
optimal, a contradiction. On the other hand, if $r_s=0$, then the vectors $[a,c]$ and $[j_s,i_s]$ are collinear. If, in
addition, $r_{s,t}=r_t=0$, then $[j_s,i_s]$ is collinear with $[j_t,-i_t]$, which, in turn, is collinear with $[b,
  -d]$. Since all four of these vectors are non-zero, we conclude that $[a,c]$ and $[b, -d]$ are collinear, and thus
$A=a\cdot d+b\cdot c =0$, a contradiction again. This completes the proof of the claim. 

Assume without loss of generality that $\Lambda$ is $y$-tight with respect to $\F$ and $[u,v]$. If $\Lambda$ is
$x$-tight with respect to $\F$ and $[u,v]$ as well, then we are done. Hence, we can also assume that $i_t \cdot a + j_t
\cdot c = w_t$ is the only (up to scaling) equality in $\I_x \cup \I_x'$. Take $a(\tau) = a+j_t\cdot \tau$,
$c(\tau)=c-i_t \cdot \tau$ and consider the area $A(\tau) = a(\tau)\cdot d+b\cdot c(\tau)$ of the fundamental
parallelogram as a linear function of $\tau$. As earlier, if the coefficient of $A(\tau)$ in front of $\tau$ were
non-zero, then $\Lambda$ would not be optimal. Thus $A(\tau)$ is a constant as a function of $\tau$.  

Now we continuously increase the value of $\tau$ from $0$ to the point $\tau_x$ where one of the strict inequalities in
$\I_x \cup \I_x'$ turns into equality for the first time; this gives us the desired second equality. Since the lattice
remains $\F$-piercing for all $0\le \tau < \tau_x$ and the set of $\F$-piercing lattices is compact, we conclude that
the new lattice is $\F$-piercing even at $\tau=\tau_x$ (even though the structure of the set $\Q$ can change at this
moment). 
 
Let $u(\tau_x) = [a(\tau_x), b]$, $v(\tau_x) = [c(\tau_x), -d]$ and $\Lambda'= \Lambda(\tau_x)$ be the optimal
$\F$-piercing lattice with the basis $[u(\tau_x), v(\tau_x)]$. Let $u'=[a',b']$, $v'=[c',-d']$ be such that $[u',v']$ is
a canonical vector basis of $\Lambda'$ constructed as in \Cref{sec:lattice}. Note that $[u',v']$ can differ from
$[u(\tau_x), v(\tau_x)]$ since the latter basis could lose each of the properties P1--P3 during the continuous
deformation process described above. To complete the proof, it remains to show that $\Lambda'$ is tight with respect to
$\F$ and $[u',v']$. 

Observe that each equality in $\I_x \cup \I_x'$ of the form $i_t \cdot a(\tau_x) + j_t \cdot c(\tau_x) = w_t$
corresponds to the equality $i_t' \cdot a' + j_t' \cdot c'=w_t$ in the new coordinates, where $i_t', j_t'$ are the
unique integers such that $i_t \cdot u(\tau_x) + j_t \cdot v(\tau_x) = i_t' \cdot u' + j_t' \cdot v'$. Since the point
$i_t \cdot u(\tau_x) + j_t \cdot v(\tau_x)$ of $\Lambda'$ is in $D$ by construction, the new coefficients $i_t', j_t'$
satisfy~\eqref{eq:ij} as well regardless of whether the construction of $[u',v']$ required a reflection of the $x$-axis
or not. (In fact, they satisfy a stronger inequality: Since the $x$-coordinate of $i_t' \cdot u' + j_t' \cdot v'$ equals
$w_t$ while the absolute value of its $y$-coordinate is at most $m_y$, the law of sines implies that $|i_t'|+|j_t'|\le
\frac{2\,\sqrt{w_t^2+m_y^2}}{\sin 22.5^\circ}$.) Therefore, $\Lambda'$ is $x$-tight with respect to $\F$ and
$[u',v']$. A similar argument shows that $\Lambda'$ is also $y$-tight with respect to $\F$ and $[u',v']$, as desired. 
\end{proof}  

We remark that not every optimal $\F$-piercing lattice $\Lambda$ is tight with respect to $\F$ and its canonical basis
$[u,v]$. For instance, if $\F$ consists of just one $1\times 1$ rectangle, $u=[1,0], v=[\tau,-1]$, then $\Lambda$ is
optimal $\F$-piercing and $y$-tight with respect to $\F$ and $[u,v]$, for all $\tau \in \RR$. However, $\Lambda$ is not
$x$-tight with respect to $\F$ and $[u,v]$ if, say, $\tau$ is irrational. 

\begin{figure}[ht]
\centering
\includegraphics[scale=0.7]{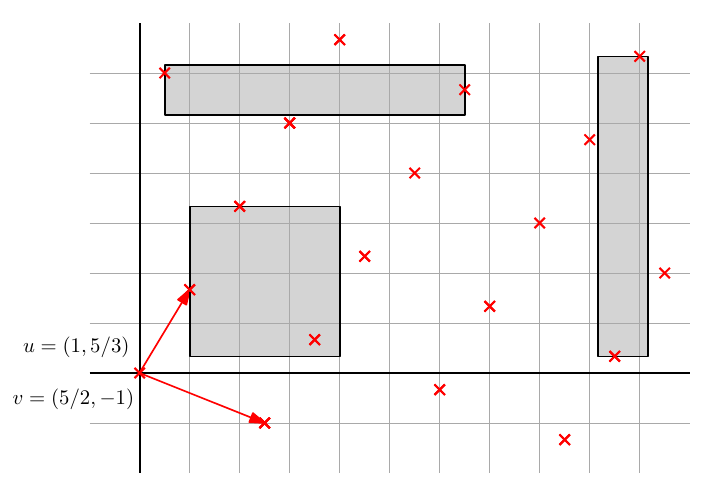}
\caption{This lattice pierces all translates of rectangles in $\F_0=\{6 \times 1, 1 \times 6, 3 \times 3\}$.}
\label{fig:lattice-F0}
\end{figure}

Lemma~\ref{lem:tight} immediately yields the following.

\begin{corollary} \label{cor:rational}
Given a family $\F=\{R_1,\dots,R_n\}$ consisting of $n$ axis-parallel rectangles with rational side-lengths,
there is an optimal $\F$-piercing lattice with a rational vector basis.   
\end{corollary}

\paragraph{Example 1.}
Let $\F_0=\{6 \times 1, 1 \times 6, 3 \times 3\}$. For the optimal $\F_0$-piercing lattice in~\cref{fig:lattice-F0},
we have $A=31/6$ and this corresponds to the pair of linear systems below.  
Their solutions are $a=1$, $c=5/2$, and $b=5/3$, $d=1$.
The above solution satisfies the inequality $b+d \leq 3$, relevant for piercing $3 \times 3$ rectangles.

\[
\begin{aligned} [t] 
\begin{cases}  
a =1 \\
a +2c =6
\end{cases}
\end{aligned}
\hspace{3cm}
\begin{aligned}[t]  
\begin{cases} 
d =1 \\
3b+d=6
\end{cases}  
\end{aligned}
\]

\paragraph{Example 2.}
A slightly suboptimal lattice for $\F_0=\{6 \times 1, 1 \times 6, 3 \times 3\}$
is given by $u=[4/5,7/4]$, $v=[13/5,-3/4]$ and has $A=103/20$. 
The corresponding linear equations are: $a+2c=6$, $c-2a=1$, $b-d=1$, $3b+d=6$.

\subsection{The algorithm}  \label{subsec:alg}

Denote by $A_{\opt}$ the area of the fundamental parallelogram of an optimal $\F$-piercing lattice. By \Cref{obs:min-area},
we have $1 \leq A_{\opt} \leq k$. 
Recall that an $\F$-piercing lattice with density $1/A'$ can be computed in $O(n)$ time~\cite{DT24}, where $A'$ satisfies
(for convenience we replaced the $1.895$-ratio from Theorem~\ref{thm:approx} by~$2$): 
\begin{equation} \label{eq:1.895}
  \frac{A_{\opt}}{2} \leq A'  \leq A_{\opt}.
\end{equation}

By Lemma~\ref{lem:tight}, there exists an optimal $\F$-piercing lattice that is tight with respect to $\F$ and a canonical basis.
The algorithm guesses an optimal lattice by exhaustive enumeration of all canonical bases $[u,v]$ such that the corresponding 
lattice is tight with respect to $\F$ and $[u,v]$ and have density in a prescribed interval.
It retains the bases corresponding to $\F$-piercing lattices and returns one with optimal density (\ie, largest area $A$). 

We have already seen, cf.~\eqref{eq:ij}, that there are $O(k)$ choices for $i,j$  and $n+2$ choices
of the right hand side for each of the four equations, so $O(k^8n^4)$ possible systems overall.
A basis is immediately discarded if $A<A'$, according to~\eqref{eq:1.895}.
For each basis satisfying suitable conditions (as imposed by P1--P3), the decision algorithm
from Section~\ref{sec:decision} determines if the corresponding lattice is $\F$-piercing.
Its running time is $O(k + n)$, and so the overall running time is
\[ O\left( k^8 n^4 (k+n) \right).  \]

We also make a couple of remarks leading to a constant factor improvement of the running time. First, observe that the coefficients $i,j$ from the proof of \Cref{lem:tight} correspond to the points of the funnel, and thus they are coprime and at least one of them is positive. Second, the distance $\mu$ between consecutive $u$-lines of an $\F$-piercing lattice cannot be too big since otherwise a translate of a rectangle in $\F$ can be placed entirely between these lines. Together with a lower bound on $A_{\opt}$, this leads to a better lower bound on $\lambda$, which, in turn, improves the range in~\eqref{eq:ij}. Finally, we remark that if we exclude the values $0$ and $m_x$ (resp., $0$ and $m_y$) from the list of possible right hand sides in the definition of an $x$-tight (resp. $y$-tight) lattice, then the statement of \Cref{lem:tight} remains valid. However, the proof of this strengthening is relatively more involved since the structure of $P$ (and thus the systems $\I_x$ and $\I_y$) can change during the continuous modification of the lattice, but only finitely many times. We do not go into the details for the sake of clarity.

\section{Separation results}  \label{sec:gap}

Theorem~\ref{thm:eps-gap}, which appears as Theorem~4 in~\cite{DT22x}, shows a preliminary separation result.
Applying our algorithm on suitable input families yields a much sharper separation,
as specified in Theorem~\ref{thm:gap-F0-F1} and Corollary~\ref{cor:gap}.
The translative and lattice piercing densities are denoted by $\den$ and $\denL$, respectively. 

\paragraph{The family $\F_0$.}
Consider the $3$-rectangle family $\F_0=\{6 \times 1, 1 \times 6, 3 \times 3\}$.
Note that $\den(\F_0)=1/6$. 
The upper bound $\denL(\F_0) \leq \frac{6}{31}$ is implied by the lattice with basis
$u=[1,5/3]$, $v=[5/2,-1]$, shown in~\cref{fig:lattice-F0}.
The program output shows that if $\Lambda$ has optimal density $\denL(\F_0)$ then $A=\frac{31}{6}$
and so  $\denL(\F_0)=1/A=6/31$; that is, there are no $\F_0$-piercing lattices with $A>31/6$. 
There are two optimal $\F_0$-piercing lattices with density $6/31$,
see Table~\ref{tab:lattices} and Fig.~\ref{fig:output1}.

\begin{table}[htbp]
\centering
{\small
  \begin{tabular}{ccccc} \toprule
$a$ & $b$ & $c$ & $d$ &  Area ($A$) \\ \midrule
$1/1$ & $5/3$ & $5/2$ & $1/1$ & $31/6=5.166\ldots$ \\ 
$5/3$ & $1/1$ & $8/3$ & $3/2$ & $31/6=5.166\ldots$ \\  
    \midrule
$1/1$ & $1/1$ & $1/1$ & $4/1$ & $5/1=5.00$ \\   
$1/1$ & $2/1$ & $1/1$ & $3/1$ & $5/1=5.00$ \\ 
\bottomrule
  \end{tabular}
}
\caption{First two rows: optimal $\F_0$-piercing lattices with density $6/31$.
Last two rows: optimal $\F_1$-piercing lattices with density $1/5$.}
\label{tab:lattices}
\end{table}

\begin{figure}[ht]
\centering
\includegraphics[scale=0.5]{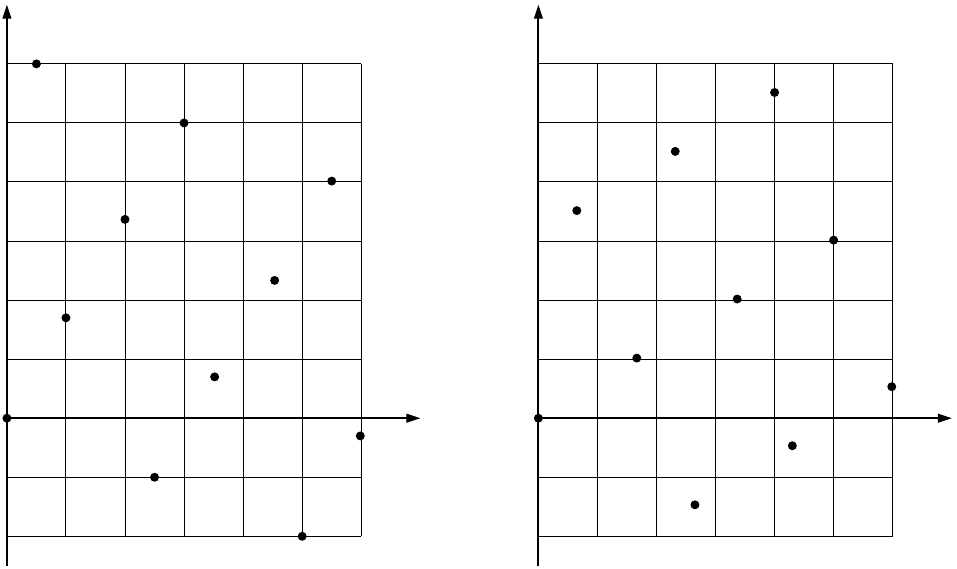}
\caption{The two optimal $\F_0$-piercing lattices of density $6/31$:
  Left: $u=[1,5/3]$, $v=[5/2,-1]$. 
  Right: $u=[5/3,1]$, $v=[8/3,-3/2]$.
  The 2nd lattice is a reflection of the 1st one about the line $y=x$.
  Since $\F_0$ is preserved under this reflection,
  the two lattices are in some sense equivalent.}
\label{fig:output1}
\end{figure}

\paragraph{The family $\F_1$.}
Consider the (extended) family $\F_1=\{6 \times 1, 1 \times 6, 3 \times 3, 4 \times 2, 2 \times 4\}$;
note that $\F_1 \supset \F_0$. It is easily verified by inspection that the periodic piercing set in~\cref{fig:F0}
is also a valid piercing set for $\F_1$; it is repeated here for convenience in Fig.~\ref{fig:non-lattice-F1}.
Note that $\den(\F_0)= \den(\F_1)=1/6$. 
  
\begin{figure}[ht]
\centering
\includegraphics[scale=1]{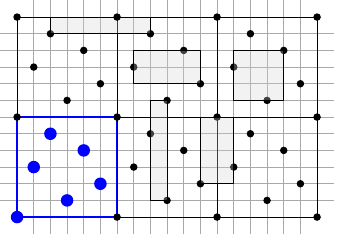}
\caption{This non-lattice periodic grid point-set pierces all translates
  of $\F_1=\{6 \times 1, 1 \times 6, 3 \times 3, 4 \times 2, 2 \times 4\}$.}
\label{fig:non-lattice-F1}
\end{figure}

None of the two optimal $\F_0$-piercing lattices pierces both rectangles in $\F_1 \setminus \F_0$;
indeed, the first lattice does not pierce the $2 \times 4$ rectangle, and the second lattice
does not pierce the $4 \times 2$ rectangle.
Since these are the only two optimal $\F_0$-piercing lattices,
we conclude that the separation gap implied by $\F_1$ is even larger.
The program finds that the optimal piercing density is $1/5$ (\ie, $1/A$ where $A=5$);
There are two optimal $\F_1$-piercing lattices with density $1/5$,
see Table~\ref{tab:lattices} and Fig.~\ref{fig:output2}, and no $\F_1$-piercing lattices with $A>5$.

\begin{figure}[ht]
\centering
\includegraphics[scale=0.5]{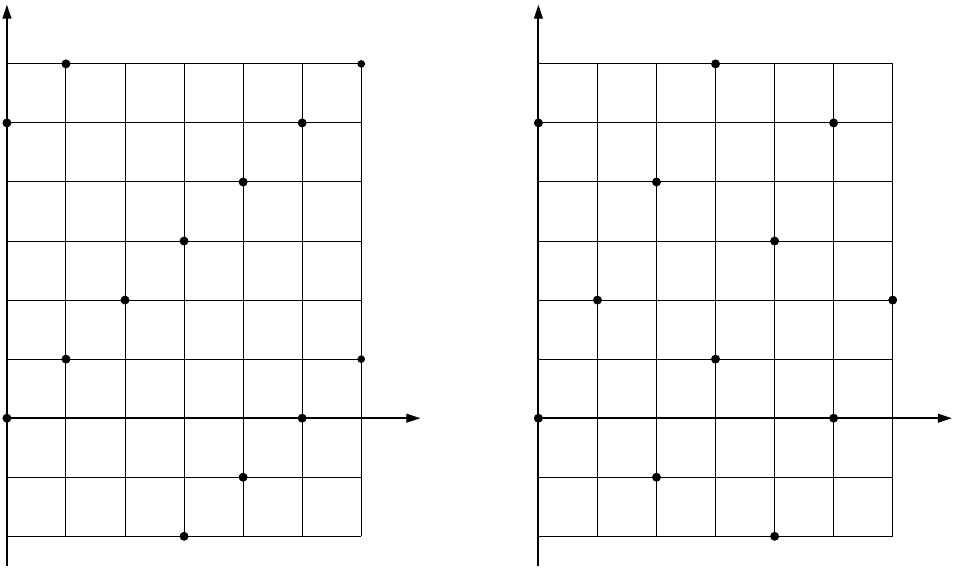}
\caption{The two optimal $\F_1$-piercing lattices of density $1/5$:
  Left: $u=[1,1]$, $v=[1,-4]$. 
  Right: $u=[1,2]$, $v=[1,-3]$.}
\label{fig:output2}
\end{figure}

\section{Conclusion}  \label{sec:conclusion}

We list several directions to be explored.

\smallskip

\begin{enumerate}
  \itemsep 1pt

\item Given a family $\F$ of $n$ axis-parallel rectangles,
  what is the computational complexity of determining the translative piercing density $\den(\F)$ 
and the  lattice piercing density $\denL(\F)$ of $\F$?
Is there a polynomial-time algorithm (independent of $k$)  for any of these problems?

\item What is the largest possible relative gap between the optimal lattice and non-lattice piercing densities
  for a family of axis-parallel rectangles? A lower bound of $1.2$ and an upper bound
  of $1.895$ on this gap are known. 

\item Is it true that $\den(\F)=\denL(\F)$ for families $\F$ of two rectangles?

\item How does the gap between the optimal lattice and non-lattice piercing densities
  for families of axis-parallel boxes grow with the dimension of the space? 
  
\item The algorithm for finding an optimal piercing lattice appears to be extendable to
  piercing axis-parallel boxes in higher dimensions. This is left as an open problem.

\end{enumerate}

\bibliographystyle{vancouver}

\end{document}